\documentclass[12pt]{article}
\usepackage{mathrsfs}
\usepackage{dsfont}
\usepackage{amsthm}
\usepackage{mathrsfs}
\usepackage{amsmath}
\usepackage{amsfonts}
\usepackage[colorlinks,linkcolor=black,anchorcolor=blue,citecolor=blue]{hyperref}
\usepackage{amssymb, amsmath, cite}

\newtheorem{theorem}{Theorem}[section]
\newtheorem{lemma}{Lemma}[section]

\newcommand\be{\begin{equation}}
\newcommand\ee{\end{equation}}
\newcommand\ber{\begin{eqnarray}}
\newcommand\eer{\end{eqnarray}}
\newcommand\berr{\begin{eqnarray*}}
\newcommand\eerr{\end{eqnarray*}}
\newcommand\bea{\begin{eqnarray}}
\newcommand\eea{\end{eqnarray}}
\newcommand\lm{\lambda}
\newcommand{\nn}{\nonumber}
\newcommand\vep{\varepsilon}\newcommand{\ii}{\mathrm{i}}
\newcommand{\dd}{\mathrm{d}}
\newcommand\e{\mathrm{e}}\newcommand\pa{\partial}
{\setlength\arraycolsep{2pt}

\begin{document}

\title{Vortices and strings induced from a generalized Abelian Higgs model}
\author{Lei Cao\\School of Mathematics and Statistics\\Henan University\\
Kaifeng, Henan 475004, PR China\\ \\Shouxin Chen\\
School of Mathematics and Statistics\\Henan University\\
Kaifeng, Henan 475004, PR China}
\date{}
\maketitle

\begin{abstract}
In this note we construct self--dual vortices and cosmic strings from the generalized Abelian Higgs theory. A special model of the theory is of focused interest in which the Higgs potential is a polynomial depending on $m$. When $|m|>0$, we obtain sharp existence theorems for vortices and strings correspond to gravity absence and appearance, respectively over the full plane. In particular, the vortex solution is unique if $m<0$. In order to over the difficulties posed by gravity, we introduce the regularization method when there are at least two distinct points among the set of centers of strings. When all these points coincide, the fixed point theorem works well. A series properties regarding vortices and strings for $|m|>0$ are also established.
\end{abstract}
\medskip
\begin{enumerate}

\item[]
{Keywords:} gauge field theory, self--duality, vortices, cosmic strings, nonlinear elliptic equations.

\item[]
{MSC numbers(2020):} 35J60, 81T13.

\end{enumerate}

\section{Introduction}\label{s0}
\setcounter{equation}{0}

Spontaneous symmetry breaking has attracted a widespread interest and is a prominent topic of research in several branches of physics, as well as in mathematics \cite{Raj,Man}. The idea behind it is essentially inspired by the thermodynamic theory of phase transitions, known as the Ginzburg--Landau theory. Quantum field describes how spontaneous symmetry breaking leads to the generation of mixed states in the form of topological defects like domain walls, vortices, monopoles, and instantons\cite{Ki,Kib}. Vortices have been studied for a long time as non--perturbative topological solutions in $2+1$ dimensions. However, the study of the effect of gravity on vortices is still in its infancy \cite{Lo,Loh}. This is partly due to the fact that Einstein's gravity is trivial in $2+1$ dimensions, in the sense that vacuum spacetime is locally flat outside of local sources although with global effects \cite{Des}. Vortices give rise to a soliton--like string structure in spacetime in the presence of gravity, and are known as cosmic strings \cite{Ze,Kib,Vi,Wi,Br}.

One of the more intriguing features of the cosmic string solution is the presence of multiple centers of energy and curvature around the local strings, which are thought to act as seeds for matter accretion in the early stages of galaxy formation. Strings are rigorously constructed only in the Bogomol'nyi or self--duality phase \cite{Lin,Co,Y,Ya,Yan,Yang}, because of the difficulty of studying the coupled system of the Einstein equations and the usual second--order Euler--Lagrange equations of the matter Lagrangian density equation.

There's a lot of research being done on cosmic strings. Brandenberger et. al. \cite{Bra} studied the classical cosmic string solution in a theory with dynamical $U(1)$ symmetry breaking. Slagter \cite{Sla} presented a cosmic string solution in Einstein--Yang--Mills Gauss--Bonnet theory on a warped $5$ dimensional spacetime conform the Randall--Sundrum--$2$ theory. Yang \cite{Yang} constructed multiple cosmic strings of an Abelian Higgs theory based on the formalism in \cite{Co}.

In this paper, we concentrate on studying another modified version of the Abelian--Higgs model \cite{Con} by coupling the Yang--Mills term of the Lagrangian to the Higgs field via a continuous function denoted by $F(|u|^2)$. From the Bogomolny argument, the model admits vortex--like topological solitons in the absence of gravity and string--like topological solitons when coupled with Einstein equations. We will present the model and its solutions which contain two types vortices and strings in the next three sections.

The rest of our paper is organized as follows. In section 2, we introduce the modified Ginzburg--Landau Lagrangian and deduce the corresponding Bogomolny equations or self--dual system which can actually be reduced to a nonlinear elliptic equation with highly complex exponential form and sources through an integration of the Einstein equations. In section 3, we establish sharp existence theorems of the model for $m>0$. In the absence of gravity, we prove the existence of vortices by using a monotone iteration. When gravity is taken into account, we introduce the method of regularization and fixed point theorem to show that the multiple string solutions are valid under a sufficient condition imposed only on the total string number $N$. In section 4, we deal with the case of the nonlinear elliptic equations which are not covered in section 3 i.e.$m<0$. We show that there is a unique vortex solution with the absence of gravity. Using a similar method as in the previous section, we can also prove the existence of string solutions, for which a sufficient condition is presented. The asymptotic estimates for the solutions obtained are established both in section 3 and in section 4.

\section{Gauge field theory and strings}\label{s1}
\setcounter{equation}{0}

It will be interesting to study the cosmic strings arising from the generalized Abelian Higgs theory. The modified Lagrangian action density assumes the form \cite{Con,Izq}
\be\label{1.1}
\mathcal{L}=-\frac{1}{4}F(s)g^{\mu\mu'}g^{\nu\nu'}F_{\mu\nu}F_{\mu'\nu'}+\frac{1}{2}g^{\mu\nu}D_\mu u\overline{D_\nu u}-V(s),
\ee
where $u$ denotes a complex scalar field, $F_{\mu\nu}=\partial_\mu A_\nu-\partial_\nu A_\mu$ is the electromagnetic curvature induced from a real--valued gauge vector field $A_\mu (\mu, \nu=0,1,2,3, t=x_0)$, $D_\mu u=\partial_\mu u-\ii A_\mu u$ is the gauge--covariant derivatives, $F(s)$ and $V(s)$ are suitable real functions depending only on $s=|u|^2$ to be determined to achieve a Bogomol'nyi structure \cite{Bog}. Let $g_{\mu\nu}$ be the gravitational metric of the spacetime with the Minkowski signature $(+---)$. With $F(s)=h^2(s)$ and $V(s)=\frac{1}{2}w^2(s)$, \eqref{1.1} becomes
\be\label{1.2}
\mathcal{L}=-\frac{1}{4}h^2(s)g^{\mu\mu'}g^{\nu\nu'}F_{\mu\nu}F_{\mu'\nu'}+\frac{1}{2}g^{\mu\nu}D_\mu u\overline{D_\nu u}-\frac{1}{2}w^2(s)
\ee
and the associated energy--momentum tensor is
\be\label{1.3}
T_{\mu\nu}=-h^2(s)g^{\mu'\nu'}F_{\mu\mu'}F_{\nu\nu'}+\frac{1}{2}(D_\mu u\overline{D_\nu u}+\overline{D_\mu u}D_\nu u)-g_{\mu\nu}\mathcal{L}.
\ee

In the following, we will consider the cosmic strings arising from the Einstein equations
\be\label{1.4}
G_{\mu\nu}=-8\pi G T_{\mu\nu}
\ee
coupled with the matter--field equation or the Euler--Lagrange equation of the matter Lagrangian density \eqref{1.2} over the gravitational spacetime.

The simplest metric element of the spacetime with the presence of gravity reads
\be\label{1.5}
\dd s^2=\dd t^2-(\dd x^3)^2-\e^\eta((\dd x^1)^2+(\dd x^2)^2),~~\eta=\eta(x^1, x^2).
\ee
If the Gauss curvature of $(\mathbb{R}^2,\e^\eta\delta_{ij})$ be denoted by $K_\eta$ which is given by
\be\label{1.6}
K_\eta=-\frac 12 \e^{-\eta}\triangle\eta.
\ee
Then, in view of \eqref{1.5}, the components of the Einstein tensor reduce into
\be\label{1.7}
-G_{00}=G_{33}=K_\eta;~~G_{\mu\nu}=0,~~(\mu,\nu)\neq(0,0)~\text{or}~(3,3).
\ee
Besides, we infer from \eqref{1.7}, the energy--momentum tensor $T_{\mu\nu}$ hold the condition
\be\label{1.8}
-T_{00}=T_{33}=-\mathcal{H};~~T_{\mu\nu}=0,~~(\mu,\nu)\neq(0,0)~\text{or}~(3,3),
\ee
where $\mathcal{H}=-\mathcal{L}$ is the Hamiltonian energy density.

With \eqref{1.7} and \eqref{1.8}, we see that the Einstein equations \eqref{1.4} are recast into the following two--dimensional single equation
\be\label{1.9}
K_\eta=8\pi G\mathcal{H}.
\ee

On the other hand, with the string metric is given by \eqref{1.5} and the complex scalar field $u$ depends only on $x^1, x^2$, we obtain from \eqref{1.3} the Hamiltonian energy density
\bea\label{1.10}
\mathcal{H}&=&T_{00}=-T_{33}=-\mathcal{L}\nn\\
&=&\frac 12\e^{-2\eta}h^2(s)F_{12}^2+\frac 12\e^{-\eta}(|D_1 u|^2+|D_2 u|^2)+\frac 12 w^2(s)\nn\\
&=&\frac 12\left((\e^{-\eta}h(s)F_{12}\mp w(s))^2+\e^{-\eta}|D_1 u\pm \ii D_2 u|^2\right)\pm \frac 12\e^{-\eta}F_{12}\nn\\
&& \pm \frac 12\e^{-\eta}\left((2h(s)w(s)-1)F_{12}+\ii(D_1 u \overline{D_2 u}-\overline{D_1 u} D_2 u)\right).
\eea
On the right hand side of \eqref{1.10}, the first term consists of quadratures, the second term is topological, and the third term may be recognized as the total divergence of a current density. For this purpose, we introduce a current density
\be\label{1.11}
J_k=\frac{\ii}{2}f(s)(u \overline{D_k u}-\overline{u} D_k u),~~k=1,2.
\ee
By Kato's identity
\be\label{1.12}
\partial_j(\phi\overline{\psi})=\overline{\psi}\partial_j\phi+\phi\partial_j\overline{\psi}=\overline{\psi}D_j\phi+\phi \overline{D_j\psi},
\ee
we can easily verify the following commutation relation for gauge--covariant derivatives
\be\label{1.13}
[D_j,D_k]u=(D_j D_k-D_k D_j)u=-\ii(\partial_j A_k-\partial_k A_j)u=-\ii F_{jk}u,~~j,k=1,2.
\ee
Differentiating \eqref{1.11} and applying $\partial_j(u\overline{u})=\overline{u}D_j u+u \overline{D_j u}$, we obtain
\be\label{1.14}
J_{12}=\partial_1 J_2-\partial_2 J_1=-sf(s)F_{12}+\ii(f(s)+f'(s)s)(D_1 u \overline{D_2 u}-\overline{D_1 u} D_2 u).
\ee
Besides, it can be shown that there holds the identity
\be\label{1.15}
|D_1 u\pm \ii D_2 u|^2=|D_1 u|^2+|D_2 u|^2\mp \ii(D_1 u \overline{D_2 u}-\overline{D_1 u} D_2 u).
\ee
If we set
\be\label{1.16}
2h(s)w(s)=1-s
\ee
and use \eqref{1.14} and \eqref{1.15}, then the right--hand side of \eqref{1.10} becomes
\be\label{1.17}
\mathcal{H}=\frac 12\left((\e^{-\eta}h(s)F_{12}\mp w(s))^2+\e^{-\eta}|D_1 u\pm \ii D_2 u|^2\right)\pm \frac 12\e^{-\eta} F_{12}\pm \frac 12\e^{-\eta} J_{12},
\ee
where $J_k$ is given by \eqref{1.11} with $f(s)=1$.

Integrating \eqref{1.17} over $(\mathbb{R}^2,\e^\eta\delta_{ij})$, we derive the Bogomol'nyi topological lower bound $E=\int \mathcal{H}\e^\eta\dd x\geq \pi N$ where $N$ represents the total string number. The energy lower bound is saturated if and only if the quadratic terms in \eqref{1.17} identically vanish:
\bea
h(|u|^2)F_{12}&=&\pm \e^\eta w(|u|^2),\label{1.18}\\
D_1 u\pm \ii D_2 u&=&0.\label{1.19}
\eea

The equation \eqref{1.19} can be rewritten as
\be\label{1.20}
(\pa_1\pm \ii \pa_2)u=\ii(A_1\pm \ii A_2)u.
\ee
Resolving \eqref{1.20}, we see that, away from the zeros of $u$, there holds the identity
\be\label{1.21}
F_{12}=\mp\frac 12 \triangle\ln|u|^2.
\ee
To proceed, we know from \eqref{1.20} and the $\overline{\pa}$--Poincar\'{e} lemma \cite{Jaf} that the zeros of $u$ are all discrete and of integer multiplicities. Let the zeros of $u$ be $p_1,p_2,\ldots,p_N$ (counting multiplicities). Then, combining \eqref{1.18} and \eqref{1.21} and observing \eqref{1.16}, we obtain the equation
\be\label{1.22}
\triangle v=\e^\eta\frac{\e^v-1}{h^2(\e^v)}+4\pi\sum_{s=1}^N \delta_{p_s}(x),~~x\in\mathbb{R}^2,
\ee
where $|u|^2=\e^v$ and $\delta_{p_s}(x)$ is the Dirac function. Besides, with \eqref{1.18} and \eqref{1.19}, we see that $T_{\mu\nu}$ defined by \eqref{1.3} satisfies $T_{\mu\nu}=0$ when $(\mu,\nu)\neq(0,0)$ or $(3,3)$. Hence, we get the Gauss curvature equation \eqref{1.9} again.

The situation without gravity, or $G=0$ or $\eta=0$ is of independent interest, and the equation \eqref{1.22} becomes the following form
\be\label{1.22b}
\triangle v=\frac{\e^v-1}{h^2(\e^v)}+4\pi\sum_{s=1}^N \delta_{p_s}(x),~~x\in\mathbb{R}^2,
\ee
which describes multiply distributed electrically and magnetically charged vortices.

We now turn our attention to the Einstein equation \eqref{1.9} where $\mathcal{H}$ is given in \eqref{1.10}. Using \eqref{1.18} to represent the second line of \eqref{1.10} as
\be\label{1.23}
\mathcal{H}=\pm\e^{-\eta}h(s)F_{12}w+\frac 12\e^{-\eta}(|D_1 u|^2+|D_2 u|^2).
\ee
When stay away from the zeros of $u$ and use the relation $|u|^2=\e^v$, we have from \eqref{1.19} that
\be\label{1.24}
|\nabla u|^2=\frac 12\e^v|\nabla v|^2.
\ee
Furthermore, we may extend \eqref{1.24} to get
\be\label{1.25}
|D_1 u|^2+|D_2 u|^2=\frac 12\e^v|\nabla v|^2.
\ee
Consequently, in view of \eqref{1.16}, \eqref{1.25} and \eqref{1.21}, we may recast \eqref{1.23} as
\be\label{1.26}
4\e^\eta\mathcal{H}=(\e^v-1)\triangle v+\e^v|\nabla v|^2,
\ee
away from the zeros $p_1,p_2,\ldots,p_N$ of $u$. Taking account of these zeros, \eqref{1.26} leads us to arrive at the relation
\be\label{1.27}
4\e^\eta\mathcal{H}=\triangle(\e^v-v)+4\pi\sum_{s=1}^N \delta_{p_s}(x),~~x\in\mathbb{R}^2.
\ee
Therefore, returning to \eqref{1.9} with \eqref{1.6}, we conclude that the quantity
\be\label{1.28}
\frac{\eta}{4\pi G}+\e^v-v+\sum_{s=1}^N\ln|x-p_s|^2
\ee
is a harmonic function over $\mathbb{R}^2$, which may be taken to be an arbitrary constant. Since we are interested in solutions in the broken symmetry category which implies
\be\label{1.29}
v(x)\to 0~~\text{as}~~|x|\to\infty.
\ee
Thus, we obtain the gravitational metric factor
\be\label{1.30}
\e^\eta=\lambda\left(\e^{v-\e^v}\prod_{s=1}^N|x-p_s|^{-2}\right)^{4\pi G}=O(|x|^{-8\pi G N}),~~\lambda>0,~~|x|\gg 1,
\ee
where $\lambda>0$ is a free coupling parameter which may be taken to be arbitrarily large. This result leads to the deficit angle
\be\label{1.31}
\delta=8\pi^2 G N.
\ee

On the other hand, integrating \eqref{1.9} leads to the total curvature
\be\label{1.32}
\int_{\mathbb{R}^2}K_\eta\e^\eta\dd x=8\pi^2 G N,
\ee
in agreement with the deficit angle \eqref{1.31} again.

Inserting \eqref{1.30} into \eqref{1.22}, we are lead to the governing equation for $N$ prescribed cosmic strings
\be\label{1.33}
\triangle v=\lambda\left(\e^{v-\e^v}\prod_{s=1}^N|x-p_s|^{-2}\right)^{4\pi G}\frac{\e^v-1}{h^2(\e^v)}+4\pi\sum_{s=1}^N \delta_{p_s}(x),~~x\in\mathbb{R}^2.
\ee

In the subsequence sections, we will concentrate on a special model of \eqref{1.2} and \eqref{1.9}. Taking
\be\label{1.34}
h(s)=\frac{\kappa}{2(1+s)^{\frac{m}{2}}},~~w(s)=\frac{(1+s)^{\frac{m}{2}}(1-s)}{\kappa},
\ee
where $\kappa>0$ and $|m|> 0$ are two parameters. In view of \eqref{1.22b} and \eqref{1.33}, this example gives us the multiple vortex equation
\be\label{1.35}
\triangle v=\beta(1+\e^v)^m(\e^v-1)+4\pi\sum_{s=1}^N \delta_{p_s}(x),~~x\in\mathbb{R}^2,
\ee
and the multiple string equation
\be\label{1.36}
\triangle v=\beta\left(\e^{v-\e^v}\prod_{s=1}^N|x-p_s|^{-2}\right)^{4\pi G}(1+\e^v)^m(\e^v-1)+4\pi\sum_{s=1}^N \delta_{p_s}(x),~~x\in\mathbb{R}^2.
\ee
where we set $\beta=\frac{4\lambda}{\kappa^2}$.

\section{Proof of existence for $m>0$}\label{s2}
\setcounter{equation}{0}

This section is devoted to study the equation \eqref{1.35} and \eqref{1.36} for $m>1$, for which we establish existence theorems for vortices and strings on the full plane and present the proof. We first study the easy case, when gravity is absent, which produces vortices. Next we introduce a $\delta$--regularization of the nonlinear elliptic equation and use an approximation method to construct strings when there are at least two distinct points among the set of centers $p_1,p_2,\ldots,p_N$. When all these points coincide, we pursue a solution by reducing the partial differential equation to its ordinary differential equation version within a radial symmetry assumption.

\subsection{Existence of vortices}

We consider the equation \eqref{1.35} for $m>1$. The main results read as follows.
\begin{theorem}\label{th3.1}
For any given prescribed vortex points $p_1, p_2,\ldots,p_N\in \mathbb{R}^2$, the equation \eqref{1.35} has a negative solution which satisfies the boundary condition \eqref{1.29} and obey the sharp decay estimates
\bea
|v(x)|&=&O(\e^{-(1-\vep)\sqrt{2^{m}\lm}|x|}),~~|x|\to\infty,\label{3.1}\\
|\nabla v(x)|&=&O(\e^{-(1-\vep)\sqrt{2^{m}\lm}|x|}),~~|x|\to\infty,\label{3.2}
\eea
where $\vep>0$ is arbitrary small.
\end{theorem}

\begin{proof}
We first show that if $v$ is a solution of \eqref{1.35}, then there must have $v(x)<0$ for all $x\in\mathbb{R}^2$. Otherwise, suppose $v(x)\geq 0$. Then the right--hand side of \eqref{1.35} says that $\triangle v\geq 0$. Consequently, the maximum principle gives us $v\leq 0$, which makes a contradiction with the assumption. Therefore, we have $v\leq 0$. Furthermore, if $v(x)=0$ for some $x\in\mathbb{R}^2$. Inserting $v(x)=0$ into \eqref{1.35} we see that, the left--hand side of \eqref{1.35} equals zero, while the right--hand side not always vanishes. Hence, $v(x)<0$ for all $x\in\mathbb{R}^2$.

In the following, we will construct super-- and subsolutions for \eqref{1.35} over $\mathbb{R}^2$. Let us consider the equation
\be\label{3.3}
\triangle v=\beta(\e^v-1)+4\pi\sum_{s=1}^N \delta_{p_s}(x),~~x\in\mathbb{R}^2, \beta>0.
\ee
Taubes \cite{Tau} has shown that the equation \eqref{3.3} has a negative solution $v_\beta$ that vanishes at infinity for all $\beta>0$ and satisfies $v_\beta<0$. Hence, \eqref{3.3} leads us to get
\bea\label{3.4}
\triangle v_\beta&=&\beta(\e^{v_\beta}-1)+4\pi\sum_{s=1}^N \delta_{p_s}(x)\nn\\
&\geq&\beta(1+\e^{v_\beta})^m(\e^{v_\beta}-1)+4\pi\sum_{s=1}^N \delta_{p_s}(x),
\eea
which means that $v_\beta$ is a subsolution of \eqref{1.35} for $m>0$ in the sense of distribution. Besides, taking $\hat{v}=0$, we see that
\be\label{3.5}
0=\triangle \hat{v}\leq\beta(\e^{\hat{v}}-1)+4\pi\sum_{s=1}^N \delta_{p_s}(x).
\ee
Hence, we conclude that $\hat{v}=0$ is a supersolution of \eqref{1.35} for $m>0$ in the sense of distribution.

Since the super-- and subsolutions are not in classical sense, we need to construct a iterative sequence in a suitable norm. For this purpose, we introduce the background functions
\be\label{3.6}
v_0=-\sum_{s=1}^N\ln\left(1+|x-p_s|^{-2}\right),~~g=4\sum_{s=1}^N(1+|x-p_s|^2)^{-2},
\ee
which satisfies
\be\label{3.7}
\triangle v_0=4\pi\sum_{s=1}^N \delta_{p_s}-g.
\ee
Use the substitution $v=v_0+w$ in \eqref{1.35}, we have
\be\label{3.8}
\triangle w=\beta(1+\e^{v_0+w})^m(\e^{v_0+w}-1)+g.
\ee
From the above analysis, we see that $w_+$ and $w_-$, which satisfy $\hat{v}=v_0+w_+$ and $v_\beta=v_0+w_-$, are super-- and subsolutions of \eqref{3.8} in the sense of distribution. Then, we can use a monotone iteration method to construct a solution $w$ of \eqref{3.8}, satisfies $w_-<w<w_+$ and vanishes at infinity. To this end, we introduce the following iteration sequence on $\mathbb{R}^2$,
\bea
\triangle w_n-k w_n&=&\beta(1+\e^{v_0+w})^m(\e^{v_0+w}-1)-kw_{n-1}+g,\label{3.9}\\
w_n&\to& 0~~\text{as}~~|x|\to\infty,~~n=2,3,\ldots,\label{3.10}\\
w_1&=&w_+,\label{3.11}
\eea
where $k>2^m\beta$. When taking $n=2$ in \eqref{3.9}, the definition of $w_1$ tells us that the right--hand side of \eqref{3.9} is an $L^2$--function on $\mathbb{R}^2$. Then the $L^2$--theory of elliptic equations imply $w_2\in W^{2,2}(\mathbb{R}^2)$. Particularly, $w_2\to 0$ as $|x|\to\infty$. By induction, we see $w_n\in W^{2,2}(\mathbb{R}^2)$ for all $n\geq 2$.

We next show that the sequence $\{w_n\}$ defined by \eqref{3.9}--\eqref{3.11} enjoys the monotone property
\be\label{3.12}
w_-<\cdots<w_n<\cdots<w_2<w_1=w_+.
\ee
Clearly, there holds $w_2<w_1=w_+$ in the neighborhood of $P=\{p_1,p_2,\ldots,p_N\}$, where we used the fact that $w_2$ is bounded. In $\mathbb{R}^2\backslash P$, we have $\triangle(w_2-w_1)\geq k(w_2-w_1)$. Thus, the maximum principle gives us $w_2<w_1$  over $\mathbb{R}^2$. Besides, we obtain from $w_-<w_+$ and $k>2^m\beta$ that
\bea\nn
(\triangle-k)(w_--w_2)&=&\left(\beta\e^{v_0+\xi}(1+\e^{v_0+\xi})^m\left(1+\frac{m(\e^{v_0+\xi}-1)}{1+\e^{v_0+\xi}}\right)-k\right)(w_--w_+)\nn\\
&\geq&(2^m\beta-k)(w_--w_+)\geq 0,~~\xi\in(w_-, w_+).\nn
\eea
Using the maximum principle again, we get $w_-<w_2$.

In general, we may assume that $w_-<w_n$ and $w_n<w_{n-1}$ for some $n\geq 2$. Then, there exists some $\xi\in(w_{n-1}, w_n)$ such that
\be\nn
(\triangle-k)(w_{n+1}-w_n)=\left(\beta\e^{v_0+\xi}(1+\e^{v_0+\xi})^m\left(1+\frac{m(\e^{v_0+\xi}-1)}{1+\e^{v_0+\xi}}\right)-k\right)(w_--w_+)\geq 0.
\ee
Hence, $w_{n+1}<w_n$. In addition, for some $\xi\in(w_-, w_n)$, we have
\be\nn
(\triangle-k)(w_--w_{n+1})=\left(\beta\e^{v_0+\xi}(1+\e^{v_0+\xi})^m\left(1+\frac{m(\e^{v_0+\xi}-1)}{1+\e^{v_0+\xi}}\right)-k\right)(w_--w_n)\geq 0.
\ee
Thus, $w_-<w_{n+1}$. Consequently, the general monotone property \eqref{3.12} holds.

Using \eqref{3.12} and a standard bootstrap argument, we see that the sequence $\{w_n\}$ converges in $W^{2,2}(\mathbb{R}^2)$. Of course, the limit $w$ is in $W^{2,2}(\mathbb{R}^2)$, and it is a classical solution of equation \eqref{3.8}.

Returning to the original variable, we get a solution $v$ of the equation \eqref{1.35} for $m>0$ which vanishes at infinity.

%Now, we show the solution to the equation \eqref{1.35} is unique. Suppose otherwise that there are two solution $v_1$ and $v_2$ for \eqref{1.35}. Then the function $u(x)=v_1(x)-v_2(x)$ satisfies the boundary condition $u(x)\to 0$ as $|x|\to\infty$ and the equation
%\be\nn
%\triangle u=\beta\e^\xi(1+\e^\xi)^m\left(1+\frac{m(\e^\xi-1)}{1+\e^\xi}\right)u,
%\ee
%where $\xi$ lies between $v_1$ and $v_2$.

Now, we pay attention to show the decay estimate of the solution $v(x)$ at infinity. The linearization of \eqref{1.35} near $v=0$ in a neighborhood of infinity is
\be\nn
\triangle v=f(x)v,
\ee
where $f(x)\to 2^m\beta$ as $|x|\to\infty$. Hence, for any $\vep\in(0, 1)$, there are some constants $C(\vep)>0$ and $R(\vep)>0$ such that
\be\label{3.13}
|v(x)|\leq C(\vep)\e^{-(1-\vep)\sqrt{2^{m}\beta}|x|},~~|x|>R(\vep),
\ee
which implies that $v$ is an $L^2$--function. If we neglect a local region, then the right--hand side of \eqref{1.35} is $L^2$, use the $L^2$--estimates in \eqref{1.35}, we see that $v\in W^{2,2}$. Fix $j=1, 2$ and set $V=\pa_j v$, differentiating \eqref{1.35} gives us
\be\label{3.14}
\triangle V=\beta m(1+\e^v)^{m-1}\e^v(\e^v-1)V+\beta (1+\e^v)^m \e^v V
\ee
away from a sufficiently large local neighborhood. Applying $v\in W^{2,2}$ and the $L^2$--estimates in \eqref{3.14}, we easily get $V\in W^{2,2}$. In particular, $V\to 0$ as $|x|\to\infty$. Asymptotical, $V$ satisfies $\triangle V=f(x)V$ where $f(x)\to 2^m\beta$ as $|x|\to\infty$. Hence $V$ also verifies the estimate \eqref{3.13} as expected.

\end{proof}

\subsection{Existence of strings}

In this subsection, we will find a solution of \eqref{1.36} satisfies $m>0$ and the boundary condition \eqref{1.29}. Since the right--hand side of \eqref{1.36} involves singularities of a power type which may cause a major obstacle. It becomes necessary to consider the orders of singularities at the strings in a few separate cases. When there are at least two distinct string centers, we use the technique of monotone iteration. While, fixed point theorem is employed when all the points among the set of centers of strings coincide.

Our existence result for multiple strings reads as follows.

\begin{theorem}\label{th3.2}
Suppose $4\pi G N\leq1$ but there are at least two distant string centers and $m>0$, then the equation \eqref{1.36} has a negative solution which vanishes at infinity.
\end{theorem}

We consider a regularized version of the equation \eqref{1.36},
\be\label{3.15}
\triangle v\!=\!\beta\left(\e^{v-\e^v}\prod_{s=1}^N(\delta+|x-p_s|^2)^{-1}\right)^{4\pi G}\!(1+\e^v)^m(\e^v-1)\!+\!\sum_{s=1}^N\frac{4\delta}{(\delta+|x-p_s|^2)^2}, m>0,
\ee
where $\delta\in(0, 1)$ is a parameter. It is worth noting that the original equation \eqref{1.36} for $m>0$ can be recovered from \eqref{3.15} by taking the $\delta\to 0$ limit.

\begin{lemma}\label{d.1}
The function $v_+=0$ is a supersolution of \eqref{3.15}.
\end{lemma}
\begin{proof}
In fact, it is clear that
\be\nn
0=\triangle v_+<\beta\left(\e^{v_+-\e^{v_+}}\prod_{s=1}^N(\delta+|x-p_s|^2)^{-1}\right)^{4\pi G}(1+\e^{v_+})^m(\e^{v_+}-1)+\sum_{s=1}^N\frac{4\delta}{(\delta+|x-p_s|^2)^2}.
\ee
Thus, $v_+=0$ is a supersolution as expected.
\end{proof}

Define
\be\label{3.16}
v_\delta=\sum_{s=1}^N\ln\left(\frac{\delta+|x-p_s|^2}{1+|x-p_s|^2}\right).
\ee
Then $v_0<v_\delta<0$, $v_\delta\to 0$ as $|x|\to\infty$ and
\be\nn
\triangle v_\delta=\sum_{s=1}^N\frac{4\delta}{(\delta+|x-p_s|^2)^2}-g,
\ee
where $g$ is defined as in \eqref{3.6}.

\begin{lemma}\label{d.2}
If $4\pi G N\leq 1$, then there is a constant $C_0>0$ independent of $0<\delta<\frac 12$ (say) such that
\be\label{3.17}
0>\beta\left(\e^{v_\delta-\e^{v_\delta}}\prod_{s=1}^N(\delta+|x-p_s|^2)^{-1}\right)^{4\pi G}(1+\e^{v_\delta})^m(\e^{v_\delta}-1)+g
\ee
holds whenever $\beta>C_0$. In other worlds, $v_-=v_\delta$ is a subsolution of \eqref{3.15} for all $\delta$.
\end{lemma}

\begin{proof}
The function $v_\delta$ can be rewrite as
\be\label{3.18}
v_\delta=-\sum_{s=1}^N\ln\left(1+\frac{1-\delta}{\delta+|x-p_s|^2}\right),
\ee
which implies that $v_\delta\to 0$ uniformly as $|x|\to\infty$. Note that
\be\nn
\e^{v_\delta}-1=\e^{\xi(v_\delta)}v_\delta,~~v_\delta<\xi(v_\delta)<0.
\ee
Thus $\e^{\xi(v_\delta)v_\delta}\to 1$ uniformly as $|x|\to\infty$. Hence
\bea\label{3.19}
&&\left(\e^{v_\delta-\e^{v_\delta}}\prod_{s=1}^N(\delta+|x-p_s|^2)^{-1}\right)^{4\pi G}(1+\e^{v_\delta})^m(\e^{v_\delta}-1)\nn\\
&\leq&\e^{4\pi G (v_\delta-1)}\prod_{s=1}^N(\delta+|x-p_s|^2)^{-4\pi G}(\e^{v_\delta}-1)\nn\\
&\leq&-\e^{4\pi G (v_\delta-1)+\xi(v_\delta)}\left(\sum_{s=1}^N\ln\left(1+\frac{1-\delta}{\delta+|x-p_s|^2}\right)\right)\prod_{s=1}^N(\delta+|x-p_s|^2)^{-4\pi G}\nn\\
&\leq&-\e^{4\pi G (v_\delta-1)+\xi(v_\delta)}\left(\sum_{s=1}^N\frac{1}{1+\xi_s}\frac{1-\delta}{\delta+|x-p_s|^2}\right)\prod_{s=1}^N(1+|x-p_s|^2)^{-4\pi G}\nn\\
&=&-h_\delta(x),
\eea
where $\xi_s\in\left(0, \frac{1-\delta}{\delta+|x-p_s|^2}\right), s=1,2,\ldots,N$.

Set $r=|x|$, according to the assumption $4\pi G N\leq 1$, we have
\be\nn
r^4 h_\delta(x)\rightarrow\infty~~\text{uniformly}~~~(\text{if}~4\pi G N< 1)~~~\text{as}~~r=|x|\rightarrow\infty
\ee
or
\be\nn
r^4 h_\delta(x)\rightarrow \text{some number}~c_\delta>0~~~(\text{if}~4\pi G N= 1)~~~\text{as}~~r=|x|\rightarrow\infty
\ee
for $0<\delta<\frac12$. Therefore, the above argument allows us to conclude that there are suitable $C_0$ and $r_0$ so that
\be\label{3.20}
-\frac{1}{r^4}(\beta r^4h_\delta-r^4g)<0,~~\forall r=|x|\geq r_0,~~\beta\geq C_0.
\ee
This shows that \eqref{3.17} holds for $r=|x|\geq r_0, \beta\geq C_0$.

On the other hand, we may choose $r_0$ sufficiently large, such that
\be\nn
P=\{p_1, p_2, \cdots, p_N\}\subset\left\{x\mid|x|<r_0\right\}\equiv \Omega.
\ee
We see by the definition of $v_\delta$ that $v_\delta<v_{1/2}~(\delta\leq\frac 12)$, thus
\be\nn
(1+\e^{v_\delta})^m(\e^{v_\delta}-1)\leq(\e^{v_{1/2}}-1),
\ee
\be\nn
\e^{4\pi G(v_\delta-\e^{v_\delta})}\prod_{s=1}^N(\delta+|x-p_s|^2)^{-4\pi G}\geq\e^{-4\pi G}\prod_{s=1}^N(1+|x-p_s|^2)^{-4\pi G}.
\ee
Hence, we can choose $C_0$ sufficiently large so that \eqref{3.17} holds on $\Omega$ as well.
\end{proof}

Since $v=0$ is a supersolution and $v_\delta<0$ is a subsolution, we can use the lemma established by Ni \cite{Ni} to derive a solution $v(x; \delta)$ of \eqref{3.15} in $\mathbb{R}^2$ satisfying
\be\label{3.21}
v_\delta\leq v(x; \delta)\leq 0,~~0<\delta<\frac 12,
\ee
and
\bea\label{3.22}
\triangle v(x; \delta)=&&\beta\left(\e^{v(x; \delta)-\e^{v(x; \delta)}}\prod_{s=1}^N(\delta+|x-p_s|^2)^{-1}\right)^{4\pi G}(\e^{v(x; \delta)}-1)\nn\\
&&+4\pi\sum_{s=1}^N\frac{\delta}{(\delta+|x-p_s|^2)^2}.
\eea
Hence, by a diagonal subsequence argument and the elliptic bootstrap iteration we can pick a sequence $\{\delta_n\}, \delta_n\to 0$ as $n\to\infty$, such that $\{v(x; \delta_n)\}$ convergence to a function $v\in C^\infty(\mathbb{R}^2\backslash P)$ in any $C^m(D)$ norm for each arbitrarily given compact subset $D$ of $\mathbb{R}^2\backslash P$ and $m\in\mathbb{N}$.

Furthermore, using the Green function to rewrite \eqref{3.22} with $\delta=\delta_n$ over an open disk containing $P$ in an integral form and setting $n\to\infty$. In view of the condition $4\pi G N\leq 1$, we see that the polynomial term on the right--hand side of \eqref{3.22} enjoys the property
\be\nn
\prod_{s=1}^N(|x-p_s|^{-8\pi G})=\left(\prod_{s=1}^N(\delta+|x-p_s|^2)^{-4\pi G}\right)\Bigg|_{\delta=0}\in L^p_{loc}(\mathbb{R}^2),~~p>1.
\ee
Letting $n\to\infty$, we can show that $v$ is a solution of the equation \eqref{1.36}. In particular, $v$ vanishes at infinity. Furthermore,
the last term in \eqref{1.36} can not be zero for all $x\in\mathbb{R}^2$, combine this fact with a maximum principle we conclude that
the solution $v$ we get from the above argument is negative on the full plane. Consequently, the proof of Theorem \ref{th3.2} is complete.

Now, we pay attention to the case that all single center of the $N$ strings is at the same point which is not covered in the restrictive condition stated in the Theorem \ref{th3.2}. Here is our existence result.

\begin{theorem}\label{th3.3}
Under the condition $4\pi G N=1$ but $N\geq1$ and $m>0$, if all the points $p_s$ are identical, that is to say, $p_s=p_0, s=1,2,\cdots, N$. Then the equation \eqref{1.36} has a symmetric solution that vanishes at infinity with the choice
\be\nn
\beta=\frac{2 N^2}{\int_{-\infty}^0 \e^{4\pi G(v-\e^v)}(1+\e^v)^m(1-\e^v)\dd v}.
\ee
\end{theorem}

\begin{proof}
Without loss of generality, we assume that $p_s=p_0, s=1,2,\cdots, N$ coincide at the origin of $\mathbb{R}^2$. For simplicity, set
\be\label{3.23}
a=4\pi G
\ee
and use $r=|x|$. Then \eqref{1.36} for $m>0$ can be written as
\be\label{3.24}
\triangle v=\beta r^{-2aN}\e^{a(v-e^v)}(1+\e^v)^m(\e^v-1)+4\pi N\delta(x).
\ee
According to the removable singularity theorem, we come up with an equivalent form of the equation \eqref{3.24},
\bea
(rv_r)_r&=&\beta r^{1-2aN}\e^{a(v-e^v)}(1+\e^v)^m(\e^v-1),\label{3.25}\\
\lim_{r\to 0}\frac{v(r)}{\ln r}&=&\lim_{r\to 0}r v_r(r)=2N.\label{3.26}
\eea

We now introduce the new variable $t=\ln r$ and use the condition $aN=1$ to rewrite the system \eqref{3.25}--\eqref{3.26} as
\bea
v''&=&\beta\e^{a(v-e^v)}(1+\e^v)^m(\e^v-1),~~-\infty<t<\infty,\label{3.27}\\
\lim_{t\to -\infty}\frac{v(t)}{t}&=&\lim_{t\to -\infty}v'(t)=2N,\label{3.28}
\eea
where $v'$ is the derivative of $v$ with respect to $t$.
Denote the right--hand side of the equation \eqref{3.27} as $h(v)$ and multiply the both side of \eqref{3.27} by $v'$, then the system \eqref{3.27}--\eqref{3.28} is recast into the integral form
\be\label{3.29}
v(t)=2Nt+\int_{-\infty}^t(t-\tau)h(v(\tau))\dd \tau,~~~t\in\mathbb{R}.
\ee
For convenient, set $w=v-2Nt$, so we can rewrite \eqref{3.29} as
\be\label{3.30}
w(t)=\int_{-\infty}^t(t-\tau)h(2N\tau+w(\tau))\dd \tau,~~~t\in\mathbb{R}.
\ee
We will show that the equation \eqref{3.30} has a solution which vanishes at $t=-\infty$. For this purpose, denote the right--hand side of the equation \eqref{3.30} as $T(w)$, then we arrive at a fixed--point problem, $w=T(w)$. Define the following function space
\be\nn
\mathcal{X}=\left\{w\in C(-\infty, t_0]~\Big|~\lim_{t\to -\infty}w(t)=0,~~\sup_{t\leq t_0}|w(t)|\leq 1\right\},
\ee
where $-\infty<t_0<\infty$ is to be determined later. First, we show that $T$ maps from $\mathcal{X}$ into itself by choosing a proper $t_0$. In fact, for any $w\in\mathcal{X}$, we have
\be\label{3.31}
|T(w)|\leq\beta\sup_{t\leq t_0}\int_{-\infty}^t|t-\tau|\e^{a(2N\tau+1)}(1+\e^{2N\tau+1})^m\dd \tau \leq 1
\ee
holds for some suitable $t_0$.
Next, we show that $T$ is a contraction operator. To achieve this goal, differentiating the right--hand side of \eqref{3.27}, we get
\be\label{3.32}
|h'(v)|=\beta\e^{a(v-e^v)}(1+\e^v)^m\left|\e^v-a(1-\e^v)^2+m\frac{\e^v-1}{1+\e^v}\right|\leq C_1\e^{av}, \forall v,
\ee
where $C_1=C_1(\beta, m)>0$. Hence, for any $w_1, w_2\in\mathcal{X}$, we have
\bea\label{3.33}
&&|T(w_1)-T(w_2)|\nn\\
=&&\left|\int_{-\infty}^t(t-\tau)h'\left(2N\tau+\widetilde{w}(\tau)\right)(w_1(\tau)-w_2(\tau))\dd \tau\right|\nn\\
\leq && C_1\sup_{t\leq t_0}|w_1(t)-w_2(t)|\int_{-\infty}^t(t_0-\tau)\e^{a(2N\tau+1)}\dd \tau,
\eea
where $\widetilde{w}(t)$ lies between $w_1(t)$ and $w_2(t)$. Therefore, when $t_0$ is properly chosen, the inequality \eqref{3.33} implies $T: \mathcal{X}\to\mathcal{X}$ is a contraction operator. Consequently, the above argument says that \eqref{3.30} has a unique solution in the neighborhood of $t=-\infty$. Besides, taking the derivative of \eqref{3.30}, we get $w'(t)\to 0$ as $t\to -\infty$ which means the boundary condition \eqref{3.31} is provided. Furthermore, we can extend $w$ to a solution over the entire $\mathbb{R}$ by using \eqref{3.30} and a standard continuation argument. Therefore, the system \eqref{3.27}--\eqref{3.28} is solved. While, in order to find a solution of \eqref{3.24}, the boundary condition
\be\label{3.34}
\lim_{t\to\infty}v(t)=0
\ee
remains to be achieved. To this end, multiplying the both side of \eqref{3.27} by $v'$ and integrating over $(-\infty, t)$, we have
\be\label{3.35}
(v'(t))^2=4N^2+2\beta\int_{-\infty}^v\e^{a(w-\e^w)}(1+\e^w)^m(\e^w-1)\dd w\triangleq 4N^2-2 H(v(t)).
\ee
It is useful to study the critical points of the equation \eqref{3.35} first. Suppose $\tilde{v}$ is a number satisfies
\be\label{3.36}
2N^2=H(\tilde{v}).
\ee
In order to ensure uniqueness at the equilibrium $\tilde{v}$, we need to require that
\be\label{3.37}
H'(\tilde{v})=\beta \e^{a(\tilde{v}-\e^{\tilde{v}})}(1+\e^{\tilde{v}})^m(1-\e^{\tilde{v}})=0.
\ee
Thus the only choice is  $\tilde{v}=0$. Inserting this result into \eqref{3.36}, we can determine the parameter $\beta$,
\be\label{3.38}
\beta=\frac{2 N^2}{\int_{-\infty}^0 \e^{a(v-\e^v)}(1+\e^v)^m(1-\e^v)\dd v}.
\ee
Although we express $\beta$ in the form of an integral, without giving it a more concrete form, we can see from the inequality
\bea
\frac{1}{a\e^a}&\leq&\int_{-\infty}^0\e^{a(v-\e^v)}(1-\e^v)\dd v\nn\\
&\leq&\int_{-\infty}^0\e^{a(v-\e^v)}(1+\e^v)^m(1-\e^v)\dd v\nn\\
&\leq& 2^m\int_{-\infty}^0\e^{a (v-\e^v)}(1-\e^v)\dd v=\frac{2^m}{a\e^a}\nn
\eea
that $\beta$ is a finite number depends on $a$ and $m$.

%Set $h(v)=\int_{-\infty}^0 \e^{4\pi G(v-\e^v)}(1+\e^v)^m(1-\e^v)$, we see that $h(v)>0$ for all $v\in(-\infty, 0)$, and $h(v)\to 0$ as $v\to -\infty$ along with $h(v)\to 0$ as $v\to 0$.

Since $v'(t)>0$ in the neighborhood of $t=-\infty$, we can rewrite \eqref{3.35} as
\be\label{3.39}
v'(t)=\sqrt{4N^2-2H(v)}\equiv\sqrt{F(v)}
\ee
and $v(t)$ is an increasing function. Besides, there holds the limit
\be\label{3.40}
\lim_{v\to 0^-}\frac{\sqrt{F(v)}}{v}=-\sqrt{\frac{F''(0)}{2}}=-\sqrt{2^m\e^{-a}\beta}<0.
\ee
Furthermore, \eqref{3.39} can be rewritten in the integral form
\be\label{3.41}
\int_{v(\tau)}^{v(t)}\frac{\dd v}{F(v)}=t-\tau.
\ee
In fact, we can deduce from \eqref{3.40}--\eqref{3.41} that $v$ satisfies the boundary condition \eqref{3.34}. To be specific, since $F(v)$ decreases monotonically with respect to $v\leq 0$ and $F(0)=0$, we have $F(v)>0$ when $v<0$. Combining this with \eqref{3.28} and \eqref{3.39}, we see that $v'(t)>0$ whenever $v<0$. In the following we will show that $v(t)<0$ for all $t$. Otherwise, if there exists some $t_0\in(-\infty, \infty)$ such that $v(t_0)=0$. Then the limit \eqref{3.40} says that there is a number $\delta>0$ such that
\be\label{3.42}
\frac{\sqrt{F(v)}}{v}<-\frac{\sqrt{2^m\e^{-a}\beta}}{2},~~t_0-\delta\leq t<t_0.
\ee
Inserting \eqref{3.42} into \eqref{3.41}, we arrive at
\be\nn
t-(t_0-\delta)>-\frac{2}{\sqrt{2^m\e^{-a}\beta}}\int_{v(t_0-\delta)}^{v(t)}\frac{\dd v}{v}=\frac{2}{\sqrt{2^m\e^{-a}\beta}}\ln\Big|\frac{v(t_0)-\delta}{v(t)}\Big|,~~t_0>t>t_0-\delta,
\ee
whose left--hand side is finite while the right--hand side approaches infinity as $t\to t_0$. This is a contradiction. Hence, $v(t)<0$ and $v'(t)>0$ for $t\in(-\infty, \infty)$. Particularly, the limit of $v(t)$ as $t\to\infty$ exists. We claim that the only possible value of this limit is $v_\infty=0$. Otherwise, assume $v_\infty<0$. Using $F(v_\infty)>0$, we see that as $t\to\infty$ the left--hand side of \eqref{3.41} is bounded and the right--hand side of \eqref{3.41} approaches infinity which makes a contradiction. Consequently, we have $v(t)\to 0$ as $t\to\infty$.

To proceed with the solution of \eqref{3.27}, let $\tau>0$ be large enough so that
\be\label{3.43}
\frac{\sqrt{F(v(t))}}{v(t)}<-(1-\vep)\sqrt{2^m\e^{-a}\beta},~~t\geq\tau,~~\vep\in(0, 1).
\ee
Inserting \eqref{3.43} into \eqref{3.41}, we get
\be\nn
t-\tau>-\frac{1}{(1-\vep)\sqrt{2^m\e^{-a}\beta}}\int_{v(\tau)}^{v(t)}\frac{\dd v}{v}=\frac{1}{(1-\vep)\sqrt{2^m\e^{-a}\beta}}\ln\bigg|\frac{v(\tau)}{v(t)}\bigg|, t>\tau,
\ee
which leads us to arrive at the decay estimate
\be\label{3.44}
|v|=O(\e^{-(1-\vep)\sqrt{2^m\e^{-a}\beta}t}),~~~t\to\infty.
\ee
Returning to the original variable $r=\e^t, u(r)=U(\ln r)$, we obtain from \eqref{3.44} that
\be\label{3.45}
|v(r)|=O(r^{-(1-\vep)\sqrt{2^m\e^{-a}\beta}}),~~~r\to\infty.
\ee
Note that $\frac{v'(t)}{|v(t)|}\to\sqrt{2^m\e^{-a}\beta}$ as $t\to\infty$,
therefore, we have
\be\label{3.46}
|v_r(r)|=O(r^{-(1-\vep)\sqrt{2^m\e^{-a}\beta}-1}),~~~r\to\infty.
\ee
The theorem \ref{th3.3} is thus proven.
\end{proof}

Next, we will give the sharp decay estimates of the string solutions obtained for equation \eqref{1.36} for $m>0$ at infinity.

Choosing $r_0>0$ sufficiently large, such that
\be\nn
P=\{p_1, p_2, \cdots, p_N\}\subset B(r_0)=\{x\in\mathbb{R}^2 \mid |x|<r_0\},
\ee
then the equation \eqref{1.36} for $m>0$ becomes
\be\label{3.47}
\triangle v=\beta\e^{a(v-\e^v)}\prod_{s=1}^N|x-p_s|^{-2a}(1+\e^v)^m(\e^v-1),~~|x|>r_0.
\ee

\begin{lemma}\label{d.3}
Suppose $a N<1$, then for any number $b>0$ there is a suitable corresponding constant $C_b>0$ such that the solution obtained in theorem \ref{th3.2} holds the bounds
\be\label{3.48}
-C_b|x|^{-b}< v(x)< 0,~~~|x|>r_0.
\ee
If $a N=1$ and there are at least two distant string centers, then \eqref{3.48} holds for $b=2$. If $a N=1$ and all the string centers coincide, then \eqref{3.48} holds for $b=(1-\vep)\sqrt{2^m\e^{-a}\beta}$.
\end{lemma}
\begin{proof}
We first show the case when $a N<1$. Introduce the comparison function
\be\label{3.49}
w(x)=C|x|^{-b}.
\ee
Then
\be\label{3.50}
\triangle w=b^2r^{-2}w,~~~|x|=r>r_0>0.
\ee
On the other hand, since $v<0$ and $v\to 0$ as $|x|\to \infty$, we have
\bea\label{3.51}
\triangle v&=&\beta\e^{a(v-\e^v)}\prod_{s=1}^N|x-p_s|^{-2a}(1+\e^v)^m(\e^v-1)\nn\\
&=&\beta\e^{a(v-\e^v)}(1+\e^v)^m\e^{\xi v}v\prod_{s=1}^N|x-p_s|^{-2a}\nn\\
&< &b^2r^{-2}v,~~~|x|=r>r_0.
\eea
where $r_0>0$ is sufficiently large and $\xi\in (0, 1)$. Hence, we have
\be\label{3.52}
\triangle(v+w)< b^2r^{-2}(v+w),~~~|x|=r>r_0.
\ee
For such fixed $r_0$, we may take $C>0$ in \eqref{3.49} large enough to make
\be\nn
(v(x)+w(x))\Big|_{|x|=r_0}> 0.
\ee
Therefore, the maximum principle says that $v\geq -w$ in $\mathbb{R}^2\backslash\overline{B(r_0)}$ as we desired.

If $a N=1$ but there are at least two distant string centers. By the definition of $v_0$, we get
\be\label{3.53}
-v_0(x)=\sum_{s=1}^N\ln\left(1+\frac{1}{|x-p_s|^2}\right)=O(r^{-2}),~~~|x|=r\to\infty,
\ee
combining with the inequalities in \eqref{3.21} and taking the $\delta\to 0$ limit, we conclude that
\be\label{3.54}
-C_2|x|^{-2}<v<0,~~~|x|\to\infty.
\ee
If $a N=1$ and all the string centers are identical, the estimate with $b=(1-\vep)\sqrt{2^m\e^{-a}\beta}$ follows from \eqref{3.45}.
\end{proof}

\begin{lemma}\label{d.4}
For the solution $v$ of \eqref{1.36} for $m>0$, there establishes for sufficiently large $r_0$ the following sharp decay estimate
\be\label{3.55}
|\nabla v|^2\leq C_b|x|^{-b},~~~|x|>r_0,
\ee
where $b>0$ is an arbitrary constant if $aN<1$ and $C_b>0$ is a constant depending on $b$. If $aN=1$ but there are at least two distant string centers, then \eqref{3.55} holds for $b=3(1-\varepsilon)$ where $\varepsilon\in(0, 1)$. If $aN=1$ and all string centers $p_s (s=1,\ldots,N)$ are identical, then $b=(1-\vep)\sqrt{2^m\e^{-a}\beta}+1$.
\end{lemma}
\begin{proof}
Since $v\in L^2(\mathbb{R}^2\backslash\overline{B(r_0)})$ and the right--hand side of \eqref{3.47} is a $L^2$--function on $\mathbb{R}^2\backslash\overline{B(r_0)}$. Using the $L^2$--theory of elliptic equations, we have $v\in W^{2,2}(\mathbb{R}^2\backslash\overline{B(r_0)})$. Set $V=\pa_j v$ and differentiating \eqref{3.47}, we have
\bea\label{3.56}
\triangle (\pa_j v)=&&\frac{\beta\e^{a(v-\e^v)}}{(1+\e^v)^m}\left(\prod_{s=1}^N|x-p_s|^{-2a}\right)
\left(\e^v-m\e^v\frac{\e^v-1}{1+\e^v}-(\e^v-1)^2\right)V\nn\\
&&+\frac{\beta\e^{a(v-\e^v)}}{(1+\e^v)^m}(\e^v-1)\pa_j\left(\prod_{s=1}^N|x-p_s|^{-2a}\right).
\eea
Because of $V\in L^2(\mathbb{R}^2\backslash\overline{B(r_0)}), j=1,2$, using the $L^2$--theory of elliptic equations again, we get $V\in W^{2,2}(\mathbb{R}^2\backslash\overline{B(r_0)}), j=1,2$, which implies $|\nabla v|$ vanishes at infinity. Consequently, there holds
\be\label{3.57}
\triangle V=H(v)r^{-2aN}V+\mathcal{O}(r^{-2aN-b-1}),
\ee
where $H(v)$ tends to some positive constant as $v\to 0$ and the meaning of exponent $b>0$ in the tail above is as given in Lemma \ref{d.3}.

Taking the comparison function
\be\label{3.58}
u(x)=C|x|^{-b_1}, C>1, b_1>1.
\ee
Then, if $aN<1$, we have
\bea\label{3.59}
\triangle u&=&b^2 r^{-2}u=(b^2+1)r^{-2}u-C r^{-2-b_1}\nn\\
&<&H(v)r^{-2aN}u-C r^{-2-b_1},~~r\geq r_0,
\eea
where $r_0$ is sufficiently large. Hence
\be\label{3.60}
\triangle(V+u)\leq H(v)r^{-2aN}(V+u)-(C r^{-2-b_1}+\mathcal{O}(r^{-2aN-b-1})),
\ee
where $b$ can be made as large as needed. Choosing suitable $b_1$ so that $2aN+b+1>2+b_1$ holds, then the inequality \eqref{3.60} gives us
\be\nn
\triangle(V+u)< H(v)r^{-2aN}(V+u),~~|x|>r_0.
\ee
Therefore, we can let $C$ in \eqref{3.58} be large enough so that $(V+u)\Big|_{|x|=r_0}\geq 0$. Consequently, the maximum principle gives us $V+u\geq 0$ for $|x|>r_0$.

If $aN=1$ but there are at least two distant string centers, we have $b=2$ in view of Lemma \ref{d.3}. Since $H(v)$ can be made arbitrarily large due to \eqref{3.21} and the adjustability of $\beta$. Thus, $V$ satisfies
\be\label{3.61}
\triangle V=H(v)r^{-2}V+\mathcal{O}(r^{-5}),
\ee
and
\be\label{3.62}
\triangle(V+u)\leq H(v)r^{-2}(V+u)-(C r^{-2-3(1-\vep)}+\mathcal{O}(r^{-5})),~~|x|>r_\vep.
\ee
Choosing $C\gg 1$ such that $C r^{-2-3(1-\vep)}+\mathcal{O}(r^{-5})>0$, then the above inequality gives us
\be\nn
\triangle(V+u)< H(v)r^{-2}(V+u).
\ee
The maximum principle again gives us $V+u\geq 0$ for $|x|>r_\vep$ when $u$ in \eqref{3.58} is sufficiently large.

On the other hand, for the case $aN<1$, we have
\be\nn
\triangle(V-u)\geq H(v)r^{-2aN}(V-u)+(C r^{-2-b_1}+\mathcal{O}(r^{-2aN-b-1})),
\ee
which means
\be\nn
\triangle(V-u)> H(v)r^{-2aN}(V-u)
\ee
holds for $2aN+b+1>2+b_1$.

When $aN=1$ but there are at least two distant string centers $p$'s. Taking $b_1=3(1-\vep)$, then for large enough $\beta$ and $r_\vep>0$, we have
\be\nn
\triangle(V-u)\geq H(v)r^{-2}(V-u)+(C r^{-2-3(1-\vep)}+\mathcal{O}(r^{-5})),~~|x|>r_\vep,
\ee
which implies $\triangle(V-u)> H(v)r^{-2}(V-u)$ as before. Consequently, in both cases, applying the maximum principle gives us $V<u$ for $|x|>r_\vep$ and large $C$ defined in \eqref{3.58}.

If $aN=1$ and all string centers $p_s, s=1,\ldots,N$ are identical, the estimate has been shown as in \eqref{3.46}.
\end{proof}

\section{Proof of existence for $m<0$}\label{s4}
\setcounter{equation}{0}

In this section, we analyse the situation $m<0$ for equations \eqref{1.35}--\eqref{1.36}. If we set $n=-m$, then the vortex equation \eqref{1.35} and the string equation \eqref{1.36} can be rewritten as
\be\label{4.1}
\triangle v=\beta\frac{(\e^v-1)}{(1+\e^v)^n}+4\pi\sum_{s=1}^N \delta_{p_s}(x),~~n>0,~x\in\mathbb{R}^2
\ee
and
\be\label{4.2}
\triangle v=\beta\left(\e^{v-\e^v}\prod_{s=1}^N|x-p_s|^{-2}\right)^{4\pi G}\frac{(\e^v-1)}{(1+\e^v)^n}+4\pi\sum_{s=1}^N \delta_{p_s}(x),~~n>0,~x\in\mathbb{R}^2,
\ee
respectively. For the case that all string centers concentrate at one point, i.e. $p_s=p_0, s=1,2,\ldots,N$. Without loss of generality, we may assume $p_0$ is the origin of the $\mathbb{R}^2$, then equation \eqref{4.2} reduces to
\be\label{4.3}
\triangle v=\beta\left(\e^{v-\e^v}r^{-2}\right)^{4\pi G}\frac{(\e^v-1)}{(1+\e^v)^n}+4\pi N \delta(x),~~n>0,~x\in\mathbb{R}^2,
\ee

Our main existence results for \eqref{4.1}--\eqref{4.3} read as follows.

\begin{theorem}\label{th4.1}
For any given prescribed vortex points $p_1, p_2,\ldots,p_N\in \mathbb{R}^2$, the equation \eqref{4.1} has a unique negative solution which vanishes at infinity and obey the sharp decay estimates
\bea
|v(x)|&=&O(\e^{-(1-\vep)\sqrt{\frac{\lm}{2^{n}}}|x|}),~~|x|\to\infty,\label{4.4}\\
|\nabla v(x)|&=&O(\e^{-(1-\vep)\sqrt{\frac{\lm}{2^{n}}}|x|}),~~|x|\to\infty,\label{4.5}
\eea
where $\vep>0$ is arbitrary small.
\end{theorem}

{\bf Remark}: Unlike the $m>0$ case studied in the previous section, for $m<0$ i.e. $n>0$, we have uniqueness of the vortex solution.

\begin{theorem}\label{th4.2}
Suppose $4\pi G N\leq1$ but there are at least two distant string centers, then the equation \eqref{4.2} has a negative solution which vanishes at infinity.
\end{theorem}

\begin{theorem}\label{th4.3}
Under the condition $4\pi G N=1$ but $N\geq1$, the equation \eqref{4.3} has a radial symmetric solution that vanishes at infinity with the choice
\be\label{4.6}
\beta=\frac{2 N^2}{\int_{-\infty}^0 \e^{4\pi G(v-\e^v)}\frac{(1-\e^v)}{(1+\e^v)^n}\dd v}.
\ee
\end{theorem}

The solutions we obtained from Theorem \ref{th4.2} and Theorem \ref{th4.3} have the following sharp decay estimates.

\begin{lemma}\label{e.3}
Suppose $a N<1$, then for any number $b>0$ there is a suitable corresponding constant $C_b>0$ such that the solution obtained by theorem \ref{th4.2} holds the bounds
\be\label{4.7}
-C_b|x|^{-b}< u(x)< 0,~~~|x|>r_0,
\ee
where $r_0$ is sufficiently large so that $\{p_1,p_2,\ldots,p_N\}\subset\{x|~|x|<r_0\}$. If $a N=1$ and there are at least two distant string centers, then \eqref{4.7} holds for $b=2$. If $a N=1$ and all the string centers coincide, then \eqref{4.7} holds for $b=(1-\vep)\sqrt{\frac{\beta}{2^{n-2}\e^a}}$.
\end{lemma}

\begin{lemma}\label{e.4}
For the solution $v$ of \eqref{1.36} with $m>0$, there establishes for sufficiently large $r_0$ the following sharp decay estimate
\be\label{4.8}
|\nabla v|^2\leq C_b|x|^{-b},~~~|x|>r_0,
\ee
where $r_0$ is sufficiently large so that $\{p_1,p_2,\ldots,p_N\}\subset\{x|~|x|<r_0\}$, $b>0$ is an arbitrary constant if $aN<1$ and $C_b>0$ is a constant depending on $b$. If $aN=1$ but there are at least two distant string centers, then \eqref{4.8} holds for $b=3(1-\varepsilon)$ where $\varepsilon\in(0, 1)$. If $aN=1$ and all string centers $p_s (s=1,\ldots,N)$ are identical, then $b=(1-\vep)\sqrt{\frac{\beta}{2^{n-2}\e^a}}+1$.
\end{lemma}

For the existence and asymptotic properties of equations \eqref{1.35}--\eqref{1.36} given above for $m<0$, its proof ideas are basically the same as the case given in the previous section for $m>0$. So we only point out the differences between the two cases and will not go into details.

(a) In the proof of Theorem \ref{th4.1}, we also consider the problem \eqref{3.3}, but taking $\beta=\beta_1\equiv\frac{\beta}{2^n}$. Then the equation \eqref{3.3} has a negative solution $v_{\beta_1}$ with $v_{\beta_1}\to 0$ as $|x|\to\infty$ and
\bea\label{4.9}
\triangle v_{\beta_1}&=&\frac{\beta}{2^n}(\e^{v_{\beta_1}}-1)+4\pi\sum_{s=1}^N \delta_{p_s}(x)\nn\\
&\geq&\beta\frac{\e^{v_{\beta_1}}-1}{(1+\e^{v_{\beta_1}})^n}+4\pi\sum_{s=1}^N \delta_{p_s}(x)
\eea
in the sense of distribution. Therefore, $\bar{v}\equiv v_{\beta_1}$ is a subsolution of \eqref{4.1} over $\mathbb{R}^2$ in the sense of distribution.
Using a similar proof of the supersolution of the equation and the monotone iteration process as in Theorem \ref{3.1}, we know that the equation \eqref{4.1} has a solution $v$ which vanishes at infinity. In fact, the solution $v$ is unique. Otherwise, let $v_1(x)$ and $v_2(x)$ are both solutions to the equation \eqref{4.1} and $U(x)=v_1(x)-v_2(x)$. Then $U(x)$ satisfies $U(x)\to 0$ as $|x|\to\infty$, and the equation
\be\nn
\triangle U=\beta\frac{\e^\xi\left(1+n(1-\e^\xi)(1+\e^\xi)^{-1}\right)}{(1+\e^\xi)^n}U,
\ee
where $\xi$ lies between $v_1(x)$ and $v_2(x)$. Then the maximum principle shows that $U=0$. Hence, we must have $v_1(x)=v_2(x)$ which implies that the equation \eqref{4.1} has a unique solution which vanishes at infinity.

(b) To prove Theorem \ref{th4.2}, we need to consider a regularized version of the equation \eqref{4.2},
\be\label{4.10}
\triangle v=\beta\left(\e^{v-\e^v}\prod_{s=1}^N(\delta+|x-p_s|^2)^{-1}\right)^{4\pi G}\!\!\frac{(\e^v-1)}{(1+\e^v)^n}+\sum_{s=1}^N\frac{4\delta}{(\delta+|x-p_s|^2)^2},
\ee
where $\delta\in(0,1)$ is a parameter. Following the proof of Lemma \ref{d.1} and Lemma \ref{d.2}, we clearly see that $v_+=0$ and $v_-=v_\delta$ are super- and subsolution of \eqref{4.10}, respectively. It is worth noting that we used inequality
\be\nn
\frac{\e^{v_\delta}-1}{(1+\e^{v_\delta})^n}\leq\frac{\e^{v_{1/2}}-1}{(1+\e^{v_{1/2}})^n}
\ee
in the proof of the subsolution of \eqref{4.10}, which is different from the enlargement technique in Lemma \ref{d.2}.

(c) The major difference between the process of proving Theorem \ref{th3.3} and Theorem \ref{th4.3} is that we take $\beta$ as
\be\nn
\beta=\frac{2 N^2}{\int_{-\infty}^0 \e^{4\pi G(v-\e^v)}\frac{(1-\e^v)}{(1+\e^v)^n}\dd v}
\ee
in Theorem \ref{th4.3}. This choosing is also reasonable, because of the bounds
\bea
\frac{1}{2^m a\e^a}&=&\int_{-\infty}^0 \e^{4\pi G(v-\e^v)}\frac{(1-\e^v)}{2^m}\dd v\nn\\
&\leq&\int_{-\infty}^0 \e^{4\pi G(v-\e^v)}\frac{(1-\e^v)}{(1+\e^v)^n\dd v}\nn\\
&\leq&\int_{-\infty}^0 \e^{4\pi G(v-\e^v)}(1-\e^v)\dd v=\frac{1}{a\e^a}.
\eea

(d) In order to obtain the asymptotic behavior described in Lemma \ref{e.3} and Lemma \ref{e.4}, we introduce the same comparison functions as in Lemma \ref{d.3} and Lemma \ref{d.4} and using the maximum principle to get conclusions.

\medskip

{\bf Acknowledgments.}
This work was supported by NSFC-12101197 and China Postdoctoral Science Foundation. No. 2022M721022.

\end{document}